\documentclass[12pt]{article}
\usepackage[left=35mm,right=35mm,top=35mm,bottom=35mm]{geometry}
\usepackage{amsmath}
\usepackage{amsthm}
\usepackage{amssymb}
\usepackage{amsfonts}
\usepackage{mathrsfs}

\DeclareMathOperator*{\slim}{s-lim}
\DeclareMathOperator*{\wlim}{w-lim}
\DeclareMathOperator{\supp}{supp}

\newcommand{\Eqn}[1]{&\hspace{-0.5em}#1\hspace{-0.5em}&}
\makeatletter

\@addtoreset{equation}{section}
\makeatother

\newtheorem{Ass}{Assumption}[section]
\newtheorem{The}{Theorem}[section]
\newtheorem{Pro}{Proposition}[section]

\title{Nonexistence of usual wave operators\\for fractional Laplacian\\and slowly decaying potentials}
\author{Atsuhide ISHIDA\\
\\
Department of Liberal Arts, Faculty of Engineering,
 \\Tokyo University of Science\\
\normalsize 3-1 Niijuku, 6-chome, Katsushika-ku,Tokyo 125-8585, Japan\\
\normalsize E-mail: aishida@rs.tus.ac.jp\\
\normalsize Fax: +81-3-5876-1616
}
\date{}

\begin{document}
\maketitle

%abstract%%%%%%%%%%%%%%%%%%%%%%%%%%
\begin{abstract}
We consider quantum systems described by the fractional powers of the negative Laplacian and the interaction potentials. When a slowly decaying potential function is given, we prove the nonexistence of the wave operators, under the assumption that the Dollard-type modified wave operators exist and that they are asymptotically complete. This nonexistence indicates the borderline between short-range and long-range behavior.
\end{abstract}

\quad\textit{Keywards}: scattering theory, wave operators, fractional Laplacian\par

%introduction%%%%%%%%%%%%%%%%%%%%%%%%
%%%%%%%%%%%%%%%%%%%%%%%%%%%%%%%
\section{Introduction\label{introduction}}
We study scattering phenomena for the fractional powers of the negative Laplacian. For $1/2<\rho\leqslant1$, the fractional power of the negative Laplacian $\omega_\rho(D)$ as a self-adjoint operator acting on $L^2(\mathbb{R}^n)$ is defined by the Fourier multiplier with the symbol
\begin{equation}
\omega_\rho(\xi)=|\xi|^{2\rho}/(2\rho),\label{free}
\end{equation}
where $D$ denotes the momentum operator $D=-i\nabla=-i(\partial_{x_1},\dots,\partial_{x_n})$. More specifically, $\omega_\rho(D)$ can be represented by the Fourier integral operator
\begin{eqnarray}
\omega_\rho(D)\phi(x)\Eqn{=}(\mathscr{F}^*\omega_\rho(\xi)\mathscr{F}\phi)(x)\nonumber\\
\Eqn{=}\int_{\mathbb{R}^n}e^{ix\cdot\xi}\omega_\rho(\xi)(\mathscr{F}\phi)(\xi)d\xi/(2\pi)^{n/2}\nonumber\\
\Eqn{=}\int_{\mathbb{R}^{2n}}e^{i(x-y)\cdot\xi}\omega_\rho(\xi)\phi(y)dyd\xi/(2\pi)^n
\end{eqnarray}
for $\phi\in\mathscr{D}(\omega_\rho(D))=H^{2\rho}(\mathbb{R}^n)$, which is the Sobolev space of order $2\rho$. In particular, when $\rho=1$, $\omega_1(D)$ is the standard free Schr\"odinger operator $\omega_1(D)=-\Delta/2=-\sum_{j=1}^n\partial_{x_j}^2/2$. Although we exclude the case where $\rho=1/2$ in this paper, $\omega_{1/2}(D)$ is the massless relativistic Schr\"odinger operator $\omega_{1/2}(D)=\sqrt{-\Delta}$.\par
We assume that the potential function $V=V(x)$ is a real-valued multiplication operator, and that its value vanishes with some decaying order when $|x|$ is sufficiently large. An apparent shape is going to be given later (Assumption \ref{ass1}). We here note that we treat the case where $V\in L^\infty(\mathbb{R})$ only. The total Hamiltonian under consideration is represented by the sum of $\omega_\rho(D)$ and $V$,
\begin{equation}
\omega_\rho(D)+V.\label{total}
\end{equation}
The operator \eqref{total} is also self-adjoint on $L^2(\mathbb{R}^n)$ because $V$ is bounded. In scattering theory, we treat the potential functions with some kind of singularity, for instance, of local Coulomb type. However, in this paper, we try to find a concrete example of a potential function such that the usual wave operators do not exist. Therefore, we do not have to consider the singularities.\par

%free_case%%%%%%%%%%%%%%%%%%%%%%%%%
In the case where the free side is the standard free Schr\"odinger operator
\begin{equation}
|D|^2/2=-\Delta/2,
\end{equation}
if we write the decay condition on $V$ as
\begin{equation}
|V(x)|\lesssim\langle x\rangle^{-\gamma}\label{polynomial_decay1}
\end{equation}
with $\gamma>0$ and $\langle x\rangle=\sqrt{1+|x|^2}$, then it is well known that, if $\gamma>1$, then the wave operators exist, and if $\gamma\leqslant1$, then the wave operators do not exist (Dollard \cite{Do2} and Reed-Simon \cite{ReSi}). That is to say, the borderline between short-range and long-range behavior is $\gamma=1$. The classical trajectory of the particle in the dynamics of the free Schr\"odinger equation has order $x(t)=O(t)$ as $t\rightarrow \infty$. The Cook-Kuroda method says that if
\begin{equation}
\int_1^\infty\|Ve^{-it|D|^2/2}\phi\|dt<\infty\label{cook_kuroda}
\end{equation}
for $\phi\in L^2(\mathbb{R}^n)$, then the wave operators
\begin{equation}
\slim_{t\rightarrow \pm \infty}e^{it(|D|^2/2+V)}e^{-it|D|^2/2}
\end{equation}
exist. This is because
\begin{eqnarray}
\lefteqn{\|(e^{it_1(|D|^2/2+V)}e^{-it_1|D|^2/2}-e^{it_2(|D|^2/2+V)}e^{-it_2|D|^2/2})\phi\|}\nonumber\\
\Eqn{\leqslant}\int_{t_2}^{t_1}\|\partial_t(e^{it(|D|^2/2+V)}e^{-it|D|^2/2})\phi\|dt=\int_{t_2}^{t_1}\|Ve^{-it|D|^2/2}\phi\|dt\longrightarrow0
\end{eqnarray}
as $t_1,t_2\rightarrow\infty$. Therefore, we can formally verify the borderline by substituting the classical order $x(t)=O(t)$ for $V(x)$,
\begin{equation}
\int_1^\infty\|Ve^{-it|D|^2/2}\phi\|dt=\int_1^\infty\|V(x(t))e^{-it|D|^2/2}\phi\|dt\lesssim\int_1^\infty t^{-\gamma}dt,\label{cook_kuroda_free}
\end{equation}
because of the decay assumption \eqref{polynomial_decay1}. This estimate is very rough and formal. However, the right-hand side of \eqref{cook_kuroda_free} is bounded if and only if $\gamma>1$.\par

%stark_case%%%%%%%%%%%%%%%%%%%%%%%%
In the case where the free side is replaced by
\begin{equation}
H_0^S=|D|^2/2-E\cdot x
\end{equation}
with $E\in\mathbb{R}^n\setminus\{0\}$, $H_0^S$ is called the Stark Hamiltonian. The classical trajectory has order $x(t)=O(t^2)$ as $t\rightarrow \infty$ by solving the Newton equation $\ddot{x}(t)=E$. By the rough estimate
\begin{equation}
\int_1^\infty\|Ve^{-itH_0^S}\phi\|dt\lesssim\int_1^\infty t^{-2\gamma}dt,
\end{equation}
we expect that the borderline will be $\gamma=1/2$. In fact, an affirmative answer to this was obtained by Ozawa \cite{Oz}. By giving a counter-example for the potential function $V$ such that the wave operators
\begin{equation}
\slim_{t\rightarrow \pm \infty}e^{it(H_0^S+V)}e^{-itH_0^S}
\end{equation}
do not exist, Ozawa \cite{Oz} determined that the borderline is $\gamma=1/2$.\par

%repulsive_case%%%%%%%%%%%%%%%%%%%%%%
Recently, Ishida \cite{Is1} found the borderline in the case where the free side is the so-called repulsive Hamiltonian
\begin{equation}
H_0^R=|D|^2/2-|x|^2/2\label{repulsive}
\end{equation}
by applying the approach developed by Ozawa \cite{Oz}. If $-|x|^2$ is replaced by $+|x|^2$ in \eqref{repulsive}, then this operator is the well-known harmonic oscillator and the particle does not scatter in this case. However, in the repulsive case, the particle can scatter with surprising velocity. Indeed, the classical trajectory is given by solving the Newton equation $\ddot{x}(t)=x(t)$, and thus we have $x(t)=O(e^{t})$ as $t\rightarrow \infty$. From our previous discussion, when we impose the decay condition on $V$ by
\begin{equation}
|V(x)|\lesssim(\log \langle x\rangle)^{-\gamma}
\end{equation}
with $\gamma>0$, it is reasonable to expect the borderline to be $\gamma=1$ because
\begin{equation}
\int_1^\infty\|Ve^{-itH_0^R}\phi\|dt\lesssim\int_1^\infty t^{-\gamma}dt.
\end{equation}
This expectation is also true. Ishida \cite{Is1} gave a counter-example for the potential $V$ such that the wave operators
\begin{equation}
\slim_{t\rightarrow \pm \infty}e^{it(H_0^R+V)}e^{-itH_0^R}
\end{equation}
do not exist.\par
We now discuss the system governed by the fractional pair $\omega_\rho(D)$ and $\omega_\rho(D)+V$. The derivative of the symbol $\omega_\rho$ at $\xi$ is
\begin{equation}
(\nabla_\xi\omega_\rho)(\xi)=|\xi|^{2\rho-2}\xi.
\end{equation}
Thus, by the canonical equation of motion, the classical trajectory of the particle follows
\begin{equation}
x(t)=|\xi|^{2\rho-2}\xi t
\end{equation}
and so has the order $x(t)=O(t)$ as $t\rightarrow\infty$. This suggests that the borderline is $\gamma=1$ when we write the decay condition on $V$ as in \eqref{polynomial_decay1}. Before presenting our main theorem, we state some assumptions.

%assumption%%%%%%%%%%%%%%%%%%%%%%%%%
\begin{Ass}\label{ass1}
Let the potential function $V=V(x)$ be a real-valued multiplication operator and have the following shape
\begin{equation}
V(x)=\lambda|x|^{-\gamma}F(|x|\geqslant1)\label{potential}
\end{equation}
with $0<\gamma\leqslant1$ and $0\not=\lambda\in\mathbb{R}$, where $F(\cdots)$ is the characteristic function of the set $\{\cdots\}$.\par
\end{Ass}

For the potential function \eqref{potential}, we define the Dollard-type modification factor $M_{\rho,\mathscr{D}}(t)$ by
\begin{equation}
M_{\rho,\mathscr{D}}(t)=\exp\left[-i\int_0^tV((\nabla_\xi\omega_\rho)(D)\tau)d\tau\right]\label{modifier}
\end{equation}
which uses the Fourier multiplier. Moreover, we write the total Hamiltonian as
\begin{equation}
H_\rho=\omega_\rho(D)+V
\end{equation}
for simplicity.

%assumption%%%%%%%%%%%%%%%%%%%%%%%%%
\begin{Ass}\label{ass2}
The Dollard-type modified wave operators
\begin{equation}
W_{\rho,\mathscr{D}}^\pm(H_\rho,\omega_\rho(D))=\slim_{t\rightarrow \pm \infty}e^{itH_\rho}e^{-it\omega_\rho(D)}M_{\rho,\mathscr{D}}(t)
\end{equation}
exist and are asymptotically complete. That is, the strong limits
\begin{equation}
W_{\rho,\mathscr{D}}^\pm(\omega_\rho(D),H_\rho)=\slim_{t\rightarrow \pm \infty}M_{\rho,\mathscr{D}}(t)^*e^{it\omega_\rho(D)}e^{-itH_\rho}P_{\rm ac}(H_\rho)
\end{equation}
exist, where $P_{\rm ac}(H_\rho)$ denotes the orthogonal projection to the absolutely continuous subspace of $H_\rho$.
\end{Ass}

The Dollard-type modified wave operators were first introduced into scattering theory by Dollard \cite{Do1, Do2} to discuss $|D|^2/2$ and the Coulomb interactions. Many details of the general long-range potentials are stated in Derezi\'{n}ski-G\'{e}rard \cite{DeGe}. It also known that this modification works well for the Stark effect (see Jensen-Yajima \cite{JeYa}, White \cite{Wh} and Adachi-Tamura \cite{AdTa}). The Dollard-type modification factor is constructed from a solution of the Hamilton-Jacobi equation. In our case, we construct this modifier by solving the equation
\begin{equation}
(\partial_tS)(t,\xi)=\omega_\rho(\xi)
\end{equation} 
and substituting $(\nabla_\xi S)(t,\xi)$ into
\begin{equation}
\exp\left[-i\int_0^tV((\nabla_\xi S)(\tau,\xi))d\tau\right].
\end{equation}
This is equal to the symbol for $M_{\rho,\mathscr{D}}(t)$ in \eqref{modifier}.
\par
We now state the main theorem of this paper, which gives an affirmative answer to the question above.

%theorem%%%%%%%%%%%%%%%%%%%%%%%%%
\begin{The}\label{the}
Under these assumptions, the usual wave operators
\begin{equation}
W_\rho^\pm(H_\rho,\omega_\rho(D))=\slim_{t\rightarrow \pm \infty}e^{itH_\rho}e^{-it\omega_\rho(D)}\label{wave_opetators}
\end{equation}
do not exist.
\end{The}

In scattering theory, it is very important to clarify the borderline between short-range and long-range behavior. In the long-range case, the usual wave operators do not exist. Therefore, we instead have to consider some kind of modified wave operator. The modification is not always unique. For the fractional powers of the negative Laplacian, Kitada \cite{Ki1,Ki2} proposed the Isozaki-Kitada-type modified wave operators and discussed their existence and completeness for $1/2\leqslant\rho\leqslant1$. Inverse scattering problems were also studied by Jung \cite{Ju} for $\rho=1/2$ and by Ishida \cite{Is2} for $1/2<\rho\leqslant1$, in the case of the short-range interactions.

%nonexistence%%%%%%%%%%%%%%%%%%%%%%%
%%%%%%%%%%%%%%%%%%%%%%%%%%%%%%%
\section{Nonexistence of Wave Operators}
We prove Theorem \ref{the} in the second section. Our proof is motivated by the original idea in Dollard \cite{Do2} (see also Reed-Simon \cite{ReSi}), which differs from the approaches in Ozawa \cite{Oz} and Ishida \cite{Is1}.\par
The following proposition yields the proof of Theorem \ref{the}

%proposition%%%%%%%%%%%%%%%%%%%%%%%%%
\begin{Pro}\label{pro}
Under Assumptions \ref{ass1} and \ref{ass2},
\begin{equation}
\wlim_{t\rightarrow \pm \infty}e^{itH_\rho}e^{-it\omega_\rho(D)}=0\label{wlim1}
\end{equation}
holds.
\end{Pro}
Before proving Proposition \ref{pro}, we first give a proof of Theorem \ref{the}.

%proof%%%%%%%%%%%%%%%%%%%%%%%%%
\begin{proof}[Proof of Theorem \ref{the}]
According to Proposition \ref{pro}, if we assume that the strong limits $W_\rho^\pm(H_\rho,\omega_\rho(D))$ in \eqref{wave_opetators} exist, then the limits have to be equal to zero. However,
\begin{equation}
\|e^{itH_\rho}e^{-it\omega_\rho(D)}\phi\|=\|\phi\|
\end{equation}
holds for any $\phi\in L^2(\mathbb{R}^n)$ because $e^{-it\omega_\rho(D)}$ and $e^{itH_\rho}$ are unitary. Therefore, $e^{itH_\rho}e^{-it\omega_\rho(D)}$ cannot converge to zero in the strong sense. This is a contradiction.
\end{proof}
We now give the proof of Proposition \ref{pro}

%proof%%%%%%%%%%%%%%%%%%%%%%%%%
\begin{proof}[Proof of Proposition \ref{pro}]
We only consider the case of $t\rightarrow\infty$ because the other case can be treated analogously. Let $\phi$ be taken from $\mathscr{S}(\mathbb{R}^n)$, which is the Schwartz functional space, such that the Fourier transform $\mathscr{F}\phi$ belongs to $C_0^\infty(\mathbb{R}^n\setminus\{0\})$ and its support satisfies
\begin{equation}
\supp\mathscr{F}\phi\subset\{ \xi\in\mathbb{R}^n\bigm||\xi|\geqslant\epsilon\}
\end{equation}
with $\epsilon>0$. By the shape of the potential \eqref{potential}, $M_{\rho,\mathscr{D}}(t)$ in \eqref{modifier} is
\begin{equation}
M_{\rho,\mathscr{D}}(t)=\exp\left[-i\lambda|D|^{-\gamma(2\rho-1)}\int_0^t\tau^{-\gamma}F(|D|^{2\rho-1}\tau\geqslant1)d\tau\right].
\end{equation}
When $|\xi|\geqslant\epsilon$ and $\tau\geqslant\epsilon^{1-2\rho}$,
\begin{equation}
F(|\xi|^{2\rho-1}\tau\geqslant1)=1
\end{equation}
holds. We thus have, under $|\xi|\geqslant\epsilon$,
\begin{gather}
\int_0^t\tau^{-\gamma}F(|\xi|^{2\rho-1}\tau\geqslant1)d\tau=\left(\int_0^{\epsilon^{1-2\rho}}+\int_{\epsilon^{1-2\rho}}^t\right)\tau^{-\gamma}F(|\xi|^{2\rho-1}\tau\geqslant1)d\tau\nonumber\\
=\int_0^{\epsilon^{1-2\rho}}\tau^{-\gamma}F(|\xi|^{2\rho-1}\tau\geqslant1)d\tau+T_{\rho,\gamma}(t),\label{eq1}
\end{gather}
where $T_{\rho,\gamma}(t)$ in \eqref{eq1} is defined by
\begin{equation}
T_{\rho,\gamma}(t)=
\begin{cases}
\ \left(t^{1-\gamma}-\epsilon^{(1-2\rho)(1-\gamma)}\right)/(1-\gamma) &\mbox{if}\quad0<\gamma<1,\\
\ \log(t/\epsilon^{1-2\rho}) &\mbox{if}\quad\gamma=1.
\end{cases}
\end{equation}
Clearly,
\begin{equation}
T_{\rho,\gamma}(t)\longrightarrow\infty\label{eq2}
\end{equation}
holds as $t\rightarrow\infty$. By \eqref{eq1}, $M_{\rho,\mathscr{D}}(t)\phi$ can be represented by
\begin{equation}
M_{\rho,\mathscr{D}}(t)\phi=\exp\left[-i\lambda T_{\rho,\gamma}(t)|D|^{-\gamma(2\rho-1)}\right]R(D)\phi,\label{eq3}
\end{equation}
where $R(D)$ is also a Fourier multiplier with its symbol
\begin{equation}
R(\xi)=\exp\left[-i\lambda|\xi|^{-\gamma(2\rho-1)}\int_0^{\epsilon^{1-2\rho}}\tau^{-\gamma}F(|\xi|^{2\rho-1}\tau\geqslant1)d\tau\right].
\end{equation}
Therefore, from this representation of $M_{\rho,\mathscr{D}}(t)\phi$, it follows that
\begin{equation}
(M_{\rho,\mathscr{D}}(t)\phi,\psi)\longrightarrow0\label{eq4}
\end{equation}
as $t\rightarrow\infty$ for any $\psi\in L^2(\mathbb{R}^n)$. Here, $(\cdot,\cdot)$ denotes the usual scalar product in $L^2(\mathbb{R}^n)$. For a more detailed explanation of limit \eqref{eq4}, note that the operator $|D|^{-\gamma(1-2\rho)}$ is self-adjoint and absolutely continuous, so we can write its spectral measure as $\mu_\rho$. By the spectral decomposition theorem and the Radon-Nikodym theorem, there exists $f_\rho=f_{\rho,\phi,\psi}\in L^1(\mathbb{R})$ such that
\begin{gather}
(M_{\rho,\mathscr{D}}(t)\phi,\psi)=\int_{-\infty}^{\infty}e^{-i\lambda T_{\rho,\gamma}(t)\sigma}d(\mu_\rho(\sigma)R(D)\phi,\psi)\nonumber\\
=\int_{-\infty}^{\infty}e^{-i\lambda T_{\rho,\gamma}(t)\sigma}f_\rho(\sigma)d\sigma=\sqrt{2\pi}(\mathscr{F}f_\rho)(\lambda T_{\rho,\gamma}(t)).\label{eq5}
\end{gather}
From \eqref{eq2}, the Reimann-Lebesgue lemma concludes that
\begin{equation}
(\mathscr{F}f_\rho)(\lambda T_{\rho,\gamma}(t))\longrightarrow0
\end{equation}
as $t\rightarrow\infty$. The convergence in \eqref{eq4} and the density argument imply that
\begin{equation}
\wlim_{t\rightarrow\infty}M_{\rho,\mathscr{D}}(t)=0.\label{wlim2}
\end{equation}
In the same way, we have
\begin{equation}
\wlim_{t\rightarrow\infty}e^{-it\omega_\rho(D)}=0,\label{wlim3}
\end{equation}
because $\omega_\rho(D)$ is self-adjoint and absolutely continuous.\par
Now, let us take any $\phi$ and $\psi$ from $L^2(\mathbb{R}^n)$. We note that the singular continuous spectrum of $H_\rho$ is empty under our assumption as proved by Kitada \cite{Ki1}. Therefore, there exist eigenvalues $E_\rho=E_{\rho,\psi}\in\mathbb{R}$ such that
\begin{equation}
H_\rho(1-P_{\rm ac}(H_\rho))\psi=E_\rho(1-P_{\rm ac}(H_\rho))\psi.\label{eigeneq}
\end{equation}
By \eqref{wlim3} and \eqref{eigeneq}, we can compute
\begin{gather}
((1-P_{\rm ac}(H_\rho))e^{itH_\rho}e^{-it\omega_\rho(D)}\phi,\psi)=(e^{-it\omega_\rho(D)}\phi,e^{-itH_\rho}(1-P_{\rm ac}(H_\rho))\psi)\nonumber\\
=e^{itE_\rho}(e^{-it\omega_\rho(D)}\phi,(1-P_{\rm ac}(H_\rho))\psi)\longrightarrow0\label{eq6}
\end{gather}
as $t\rightarrow\infty$. On the other hand, 
\begin{eqnarray}
\lefteqn{(P_{\rm ac}(H_\rho)e^{itH_\rho}e^{-it\omega_\rho(D)}\phi,\psi)=(M_{\rho,\mathscr{D}}(t)\phi,M_{\rho,\mathscr{D}}(t)^*P_{\rm ac}(H_\rho)e^{it\omega_\rho(D)}e^{-itH_\rho}\psi)}\nonumber\\
\Eqn{=}(M_{\rho,\mathscr{D}}(t)\phi,W_{\rho,\mathscr{D}}^\pm(\omega_\rho(D),H_\rho)\psi)\nonumber\\
\Eqn{}+(M_{\rho,\mathscr{D}}(t)\phi,(M_{\rho,\mathscr{D}}(t)^*P_{\rm ac}(H_\rho)e^{it\omega_\rho(D)}e^{-itH_\rho}-W_{\rho,\mathscr{D}}^\pm(\omega_\rho(D),H_\rho))\psi).\qquad\label{eq7}
\end{eqnarray}
The first term of the right-hand side of \eqref{eq7} goes to zero by \eqref{wlim2} as $t\rightarrow\infty$. By the Schwarz inequality and the unitariness of $M_{\rho,\mathscr{D}}(t)$, the absolute value of the second term is
\begin{eqnarray}
\lefteqn{|(M_{\rho,\mathscr{D}}(t)\phi,(M_{\rho,\mathscr{D}}(t)^*P_{\rm ac}(H_\rho)e^{it\omega_\rho(D)}e^{-itH_\rho}-W_{\rho,\mathscr{D}}^\pm(\omega_\rho(D),H_\rho))\psi)|}\nonumber\\
\Eqn{\leqslant}\|\phi\|\|(M_{\rho,\mathscr{D}}(t)^*P_{\rm ac}(H_\rho)e^{it\omega_\rho(D)}e^{-itH_\rho}-W_{\rho,\mathscr{D}}^\pm(\omega_\rho(D),H_\rho))\psi\|,
\end{eqnarray}
where $\|\cdot\|$ denotes the usual $L^2$-norm, which also goes to zero as $t\rightarrow\infty$ by Assumption \ref{ass2}. We thus obtain, as $t\rightarrow\infty$,
\begin{equation}
(P_{\rm ac}(H_\rho)e^{itH_\rho}e^{-it\omega_\rho(D)}\phi,\psi)\longrightarrow0.\label{eq8}
\end{equation}
Combining \eqref{eq6} and \eqref{eq8} shows that \eqref{wlim1} holds. This completes the proof.
\end{proof}

We have to exclude the case where $\rho=1/2$ because \eqref{wlim2} does not always hold in this case. When $\rho=1/2$, $M_{1/2,\mathscr{D}}(t)$ has the following shape
\begin{equation}
M_{1/2,\mathscr{D}}(t)=\exp\left[-i\lambda T_{1/2,\gamma}(t)\right].
\end{equation}
This is a scalar and oscillates as $t\rightarrow\infty$. As stated in the first section, the operator $\omega_{1/2}(D)=\sqrt{-\Delta}$ describes a relativistic system and so there are large differences between the cases $\rho=1/2$ and $1/2<\rho\leqslant1$.\\

%acknoledgemants%%%%%%%%%%%%%%%%%%%%%
\noindent\textbf{Acknowledgments.} This work was partially supported by the Grant-in-Aid for Young Scientists (B) \#16K17633 from JSPS.

%references%%%%%%%%%%%%%%%%%%%%%%%%%
%%%%%%%%%%%%%%%%%%%%%%%%%%%%%%%


\begin{thebibliography}{99}
\bibitem{AdTa}T.\ Adachi, H.\ Tamura, Asymptotic completeness for long-range many particle systems with Stark effect II, \textit{Comm.\ Math.\ Phys.} \textbf{174} (1996) 537--559.
\bibitem{Do1}J.\ D.\ Dollard, Asymptotic convergence and the Coulomb interaction, \textit{J.\ Mathematical Phys.} \textbf{5} (1964)  739--731.
\bibitem{Do2}J.\ D.\ Dollard, Quantum-mechanical scattering theory for short-range and Coulomb interactions, \textit{Rocky Mountain J.\ Math.} \textbf{1} (1971) 5--81.
\bibitem{DeGe}J.\ Derezi\'{n}ski, C.\ G\'{e}rard, \textit{Scattering Theory of Classical and Quantum N-Particle Systems}, Springer-Verlag, 1997.
\bibitem{Is1}A.\ Ishida, The borderline of the short-range condition for the repulsive Hamiltonian, \textit{J.\ Math.\ Anal.\ Appl.} \textbf{438} (2016) 267--273.
\bibitem{Is2}A.\ Ishida, Propagation property and its application to inverse scattering for fractional powers of the negative Laplacian, \textit{axXiv} 1612.01683.
\bibitem{Ju}W.\ Jung, Geometrical approach to inverse scattering for the Dirac equation, \textit{J.\ Math.\ Phys.}\ \textbf{38} (1997) 39--48.
\bibitem{Ki1}H.\ Kitada, Scattering theory for the fractional power of negative Laplacian, \textit{Jour. Abstr. Differ. Equ. Appl.}\ \textbf{1} (2010) 1--26.
\bibitem{Ki2}H.\ Kitada, A remark on simple scattering theory, \textit{Commun.\ Math.\ Anal.}\ \textbf{11} (2011) 123--138.
\bibitem{Oz}T.\ Ozawa, Non-existence of wave operators for Stark effect Hamiltonians, \textit{Math.\ Z.} \textbf{207}(1991) 335--319.
\bibitem{ReSi}M.\ Reed, B.\ Simon, \textit{Methods of Modern Mathematical Physics III, Scattering Theory}, Academic Press, New York, 1979.
\bibitem{JeYa}A.\ Jensen, K.\ Yajima, On the long-range scattering for Stark Hamiltonians, \textit{J.\ Reine Angrew.\ Math.} \textbf{420} (1991) 179--193.
\bibitem{Wh}D.\ White, Modified wave operators and Stark Hamiltonians, \textit{Duke Math.\ J.\ } \textbf{68} (1992) 83--100.
\end{thebibliography}
\end{document}